\documentclass{article}

\newcommand{\citet}[1]{cite{#1}}

\usepackage{amsthm, amsmath, amssymb}
\usepackage{graphicx} 
\usepackage{mathtools}
\usepackage{etoolbox}
\usepackage[capitalize]{cleveref}
\usepackage{xcolor}
\usepackage{comment}
\usepackage{enumerate}
\usepackage{multicol}
\usepackage[normalem]{ulem}
\usepackage{xspace}
\usepackage{skak}

\usepackage[format=plain, labelfont=it, textfont=it]{caption}

\theoremstyle{plain}
\newtheorem{theorem}{Theorem}
\newtheorem{innercustomthm}{Theorem}
\newenvironment{customthm}[1]
  {\renewcommand\theinnercustomthm{#1}\innercustomthm}
  {\endinnercustomthm}  

\newtheorem{question}[theorem]{Question}

\newtheorem{corollary}[theorem]{Corollary}

\newtheorem{lemma}[theorem]{Lemma}
\newtheorem{innercustomlem}{Lemma}
\newenvironment{customlem}[1]
  {\renewcommand\theinnercustomlem{#1}\innercustomlem}
  {\endinnercustomlem}

\newtheorem{observation}[theorem]{Observation}

\newtheorem{conjecture}[theorem]{Conjecture}

\newtheorem{proposition}[theorem]{Proposition}

\theoremstyle{definition}
\newtheorem{definition}[theorem]{Definition}
\newtheorem{innercustomdef}{Definition}
\newenvironment{customdef}[1]
  {\renewcommand\theinnercustomdef{#1}\innercustomdef}
  {\endinnercustomdef}

\newtheorem{example}{Example}

\theoremstyle{remark}

\renewcommand{\vec}[1]{\bar{#1}}

\newcommand{\FormatRuleClass}[1]{\ensuremath{\mathtt{#1}}\xspace}
\newcommand{\fus}{\FormatRuleClass{fus}}
\newcommand{\bdd}{\FormatRuleClass{bdd}}
\newcommand{\fc}{\FormatRuleClass{fc}}
\newcommand{\fcbdd}{(\bdd\ensuremath{\Rightarrow}\fc)\xspace}

\newcommand{\Exists}[1]{\exists{\scalebox{0.9}{$#1$}}\;}
\newcommand{\Forall}[1]{\forall{\scalebox{0.9}{$#1$}}\;}

\newcommand{\cq}{\ensuremath{q}}
\newcommand{\ucq}{\ensuremath{Q}}

\newcommand{\function}[1]{\mathsf{#1}}

\newcommand{\hfun}{\function{h}}
\renewcommand{\hom}{\hfun}

\newcommand{\predicate}[1]{\mathtt{#1}}
\newcommand{\ppred}{\predicate{P}}

\newcommand{\apred}{\predicate{A}}
\newcommand{\bpred}{\predicate{B}}

\newcommand{\dpred}{\predicate{D}}
\newcommand{\epred}{\predicate{E}}

\newcommand{\xpred}{\predicate{X}}
\newcommand{\ypred}{\predicate{Y}}

\newcommand{\signature}{\mathbb{S}}
\newcommand{\sig}{\signature}
\newcommand{\arity}[1]{\function{ar}(#1)}

\newcommand{\database}{\instance}
\newcommand{\db}{\database}
\newcommand{\structure}[1]{\mathcal{#1}}
\newcommand{\instance}{\structure{I}}
\newcommand{\inst}{\instance}
\newcommand{\jnstance}{\structure{J}}
\newcommand{\jnst}{\jnstance}
\newcommand{\adom}[1]{\function{adom}(#1)}

\newcommand{\query}{Q}

\newcommand{\knowledgebase}{\db, \rs}
\newcommand{\kb}{\knowledgebase}

\newcommand{\set}[1]{\{\,#1\,\}}
\newcommand{\singleton}[1]{\{#1\}}
\newcommand{\pair}[1]{\langle\, #1 \,\rangle}
\newcommand{\spair}[1]{\langle #1 \rangle}

\newcommand{\vx}{\vec{x}}

\newcommand{\vy}{\vec{y}}
\newcommand{\vvv}{\vec{v}}

\newcommand{\vz}{\vec{z}}

\newcommand{\vt}{\vec{t}}

\newcommand{\va}{\vec{a}}

\newcommand{\fulltrig}{\pair{\eru, \hom}}
\newcommand{\trig}{\tau}
\newcommand{\trigoutput}[1]{\function{output}(#1)}
\newcommand{\existentialrule}{\rho}
\newcommand{\eru}{\existentialrule}
\newcommand{\chasesymbol}{Ch}

\newcommand{\chase}[1]{\chasesymbol(#1)}
\newcommand{\step}[2]{\chasesymbol_{#1}(#2)}

\newcounter{rscounter}
\newcounter{dbcounter}
\newcommand{\ruleset}{\mathcal{R}}
\newcommand{\rs}{\ruleset}
\newcommand{\rsb}{\mathcal{S}}
\newcommand{\rsbb}{\mathcal{S'}}

\newcommand{\iffi}{\textit{iff} }

\newcommand{\fr}[1]{\function{fr}(#1)}

\newcommand{\exvars}[1]{\function{\ensuremath{\exists}vars}(#1)}

\newcommand{\head}[1]{\function{head}(#1)}
\newcommand{\body}[1]{\function{body}(#1)}

\definecolor{mint}{rgb}{0, 0.8, 0.7}
\definecolor{lilac}{rgb}{0.8, 0.4, 0.7}
\definecolor{fox}{rgb}{0.9, 0.5, 0.3}

\newcommand{\piotr}[1]{{\small \textbf{P:} \color{mint}#1}}

\newcommand{\nats}{\mathbb{N}}

\newcommand{\cnarule}[3]{#1 \,\to\,\Exists{#2}\; #3}
\newcommand{\narule}[3]{#1 \;\;\to\;\;\Exists{#2}\; #3}
\newcommand{\nadrule}[2]{#1 \;\;\to\;\;#2}

\newcommand{\rew}{\mathtt{rew}}

\newcommand{\homequiv}{\leftrightarrow}

\newcommand{\arbtour}[3]{\textup{\textsc{Tournaments}}_{#3}}
\newcommand{\arbtoure}{\arbtour{}{}{\epred}}
\newcommand{\loopq}{\textup{\textup{\textsc{Loop}}}_{\epred}}
\newcommand{\mainstatement}[2]{\chase{#1, #2} \models \arbtoure \quad \Rightarrow \quad \chase{#1, #2} \models \loopq}
\newcommand{\mainstatementtopinstance}[2]{\chase{#2} \models \arbtoure \quad \Rightarrow \quad \chase{#2} \models \loopq}

\newcommand{\topinst}{\{\top\}}

\newcommand{\largecrown}{\includegraphics[scale=.035]{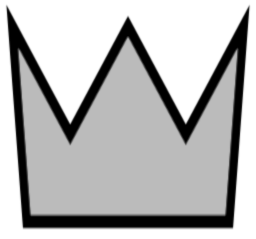}}
\newcommand{\crown}{\includegraphics[scale=.021]{img/crown.png}}
\renewcommand{\largecrown}{\includegraphics[scale=.11]{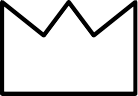}}
\renewcommand{\crown}{\includegraphics[scale=.07]{img/crown-alternate-alternate.pdf}}

\newcommand{\chasec}{\chase{\rsc}}

\newcommand{\chasece}{\chase{\rsce}}

\newcommand{\rscrown}{\rs_{\crown}}
\newcommand{\rsc}{\rscrown}

\newcommand{\rscrownexists}{\rs_{\crown}^{\exists}}
\newcommand{\rsce}{\rscrownexists}

\newcommand{\rscrowndatalog}{\rs_{\crown}^{DL}}
\newcommand{\rscdl}{\rscrowndatalog}

\newcommand{\qcrown}{\ucq_{\crown}}
\newcommand{\qc}{\qcrown}

\newcommand{\timestamp}[1]{\mathtt{TS}(#1)}
\newcommand{\timestamps}[1]{\mathtt{TS_{\mathsf{m}}}(#1)}

\newcommand{\multiset}[1]{\{\,#1\,\}_{\mathsf{m}}}
\newcommand{\memptyset}{\emptyset_{\mathsf{m}}}
\newcommand{\mmax}{\max_{\mathsf{m}}}
\newcommand{\mcup}{\cup_{\mathsf{m}}}
\newcommand{\mcap}{\cap_{\mathsf{m}}}
\newcommand{\msetminus}{\setminus_{\mathsf{m}}}
\newcommand{\lexorder}{\leq_\mathrm{lex}}
\newcommand{\slexorder}{<_\mathrm{lex}}

\newcommand{\modelsinj}{\models_\mathrm{inj}}
\newcommand{\wn}[1]{{\mathcal{W}(#1)}}

\newcommand{\hompi}{\function{\ensuremath{\pi}}}
\newcommand{\fulltrigpi}{\pair{\eru, \hompi}}

\newcommand{\disjunion}{\;\bar{\cup}\;}

\newcommand{\reify}{\function{reify}}

\author{Lucas Larroque\\
{\small \em Inria, DI ENS, ENS, CNRS, PSL University}\\
{\small \texttt{lucas.larroque@inria.fr}}\\ \phantom{a}\\
Piotr Ostropolski-Nalewaja\\
{\small \em University of Wrocław}\\
{\small \texttt{postropolski@cs.uni.wroc.pl}}\\ \phantom{a}\\
Michaël Thomazo\\
{\small \em Inria, DI ENS, ENS, CNRS, PSL University}\\
{\small \texttt{michael.thomazo@inria.fr}}}

\begin{document}

\date{}

\title{No Cliques Allowed:\\ The Next Step Towards BDD/FC Conjecture}

\maketitle
\begin{abstract}
This paper addresses one of the fundamental open questions in the realm of existential rules: the conjecture on the finite controllability of bounded derivation depth rule sets \fcbdd. We take a step toward a positive resolution of this conjecture by demonstrating that universal models generated by \bdd rule sets cannot contain arbitrarily large tournaments (arbitrarily directed cliques) without entailing a loop query, $\exists{x}~E(x,x)$. This simple yet elegant result narrows the space of potential counterexamples to the \fcbdd conjecture.
\end{abstract}
\section{Introduction}
The last decade has witnessed a fruitful interplay between the database and the knowledge representation communities, aiming to facilitate the integration and querying of legacy databases. The general idea is to allow a user to formulate queries using a logical formalism, whose predicates' meaning is constrained through first-order formulas. This leads to a core reasoning problem, coined \emph{ontology based query answering} (OBQA), that takes as input a database $\instance$, a rule set $\rs$ and a Boolean conjunctive query $\cq$, and asks whether $\instance, \rs \models \cq$, where $\models$ denotes the entailment relation of first-order logic. In other terms, is that true that any model (finite or infinite) of $\database$ and $\rs$ is a model of $\cq$? Rules are often expressed as \emph{existential rules}\footnote{Existential rules are also referred to as a \emph{tuple-generating dependencies}~\cite{AbiteboulHV95}, \emph{conceptual graph rules}~\cite{SalMug96}, Datalog$^\pm$~\cite{Gottlob09}, and \emph{$\forall \exists$-rules}~\cite{BagLecMugSal11} in the literature.} which are formulas of the shape $\forall \vx,\vy~ B(\vx, \vy) \rightarrow \exists \vz~ H(\vec{y}, \vec{z})$, where $B$ and $H$ are conjunctions of atoms respectively called the \emph{body} and the \emph{head} of the rule. It is well known that when no further constraints are put on rules, OBQA is an undecidable problem. 

\paragraph{Decidability of OBQA}
A significant research effort has thus been devoted to design conditions on $\rs$ that ensure the decidability and sometimes the tractability of Boolean conjunctive query answering under $\rs$. An important tool to understand these conditions is the \emph{chase}~\cite{chase-first-paper-maier-mend-sag-79}. In a nutshell, the chase is an algorithm that adds fresh terms and new atoms to ensure that each mapping of a rule body can be extended into a mapping of a rule head. The result of this possibly infinite process is a specific model of $\instance$ and $\rs$, which has a universal property~\cite{chase-revisited}: it homomorphically maps to any model of $\instance$ and $\rs$. When the chase is finite (which is for instance the case when the rule set fulfills some kind of acyclicity \cite{chase-intro-fagin-kola-mil-popa-05,krotzsch-rudolph-11-jointly-acycclic,CuencaGrau13,GerlachCarral23,Carral17}), this provides a decision procedure for query answering. It is also the case when the chase is not finite, but is a structure of bounded treewidth --- this is typically the case for guarded rules~\cite{Cali-GK:guarded-weakly-guarded-tgd-introduced}. 

Finally, and of great interest for our work, is the case of UCQ-rewritability. A rule set $\rs$ is said to be UCQ-rewritable if conjunctive queries (CQs) can always be rewritten. Specifically, for every CQ $\cq$, there exists a union of conjunctive queries (UCQ) $\ucq_\rs$ such that, for any database $\db$:  
\[
\db, \rs \models \cq \;\;\iff\;\; \db \models \ucq_\rs.
\]  
Given that the chase can produce an infinite structure, the UCQ-rewritability of a rule set ensures decidable OBQA, as the rewriting can be computed and subsequently used to query the database. Unsurprisingly, UCQ-rewritability has been the subject of extensive research, leading to the introduction of numerous subclasses over the past few decades, including:  
- \emph{Linear theories}, which permit at most a single atom in rule bodies~\cite{oblivious_chase_modern}; \emph{Guarded \bdd theories}, which generalize linear theories~\cite{BBLP18,CR15}; \emph{Sticky theories}, where strong restrictions are imposed on joins~\cite{sticky-intro}.

UCQ-rewritability is also studied independently under a number of different names --- highlighting the fundamental nature of the property --- with notable cases of finite unification sets class (\fus) and bounded derivation depth property (\bdd). The \fus class is closely tied to the concept of the backward chaining procedure, which is essentially a process of ``chasing in reverse.'' In this approach, to answer the entailment question, one starts with a query and attempts to derive a substructure of the database by applying rules in reverse. The \bdd class, formally introduced in the next section, has a much older origin, dating back to the 1980s. It stems from the classical notion of the boundedness of a Datalog program~\cite{GMSV87}.


\paragraph{OBQA in the finite}
Databases, however, are finite structures, and under the understanding that $\instance$ is a partial description of the real world, and that $\rs$ are constraints that should hold on the finite real world, the classical notion of entailment provided by first-order logic does not correspond to the desired notion of entailment. Indeed, one should be more interested in knowing whether ``q holds in any \emph{finite} database that contains $\instance$ and is a model of $\rs$''. This semantics is clearly related with the previous one, as a consequence of $\db$ and $\rs$ in the unrestricted semantics is also one in the finite one. The converse is however not true, as the following prototypical example witnesses.

\begin{example}
\label{ex-not-fc}
    Let us consider $\db = \{\epred(a,b)\}$, and $\rs$ containing the two rules $\forall x\forall y ~ \epred(x,y) \rightarrow \exists z~\epred(y,z)$ and  $\forall x \forall y \forall z~ \epred(x,y) \wedge \epred(y,z) \rightarrow \epred(x,z)$. The query $\exists x~ \epred(x,x)$ is not a consequence of $\kb$ under the unrestricted semantics, as witnessed by the chase. In any finite model however, there must be a cycle, and hence a loop because of the transitivity enforced by the second rule. 
\end{example}

\paragraph{Finite Controllability}
Finite reasoning provides a user with a semi-decision procedure for non-entailment (enumerating finite interpretations, and checking whether there exists one that is not a model of $\query$), but entailment is not semi-decidable anymore, as a finite universal model may not exist. An interesting question, already raised by Rosati \cite{DBLP:journals/jcss/Rosati11}, is, given a rule set $\rs$, whether for any database $\db$ and any query $\query$ the unrestricted and the finite semantics coincide. If so, the rule set $\rs$ is said to be \emph{finitely controllable} (\fc). Note that whenever a rule set is \fc, query answering becomes decidable. Indeed, the classical mathematical tools for semi-deciding first order entailment are still available, while non-entailment can be witnessed by finite structures. Rosati showed that inclusion dependencies, a very restricted form of existential rules, enjoy finite controllability~\cite{DBLP:journals/jcss/Rosati11}. This has been generalized in two ways, for guarded rules~\cite{guarded_is_fc}, as well as for sticky rules~\cite{sticky_is_fc}. Note that both classes generalize inclusion dependencies in various ways, guarded rules enjoying the bounded treewidth property, while sticky rules enjoy the bounded derivation depth property.

There is a well-known conjecture by Gogacz and Marcinkowski~\cite{bdd_fc_conjecture} stating that rule sets with the bounded derivation depth (\bdd) property are finitely controllable \fcbdd. Moreover, they demonstrated that the conjecture holds in a specific restricted setting where \bdd rule sets are defined over a binary signature and have heads containing only a single atom~\cite{bdd_fc_conjecture}. However, this result applies to a highly limited case, and little progress has been made towards resolving the conjecture in its general form. The broader question of whether finite controllability extends to all \bdd rule sets remains an important open problem. 

\paragraph{Contributions}
Our main contribution is to show that \bdd rule sets enjoy certain model-theoretic property. We write $\arbtour{}{}{\epred}$ to denote the query: \emph{for all integers $k$, there is a set $\set{x_1,\dots,x_k}$ of elements in the given instance such that $\epred(x_i,x_j)$ or $\epred(x_j,x_i)$ holds in the instance for all $i\neq j$}, and write $\loopq$ to denote the query $\exists{x}~\epred(x,x)$. Then, the following is our main result:
\begin{theorem}\label{thm:main}
For every \bdd rule set $\rs$ and every instance $\instance$ we have: 
\begin{align*}
(\inst, \rs) \;\models\; \arbtour{}{}{\epred} \quad\Rightarrow\quad (\inst,\rs) \;\models\; \loopq. \tag{\textup{\sympawn}}\label{tag:thm:main}
\end{align*}
\end{theorem}
%
In addition to providing insights into the structural properties of universal models of \bdd rule sets, this result serves as an important step towards proving the \fcbdd conjecture. Specifically, it narrows the space of potential counterexamples to the conjecture by eliminating the most natural ones. We briefly explain why this is the case.

 First note that any model that entails $\arbtour{}{}{\epred}$ and not $\loopq$ is infinite, as the terms $x_1,\dots,x_k$ in $\arbtour{}{}{\epred}$ must be distinct for $\loopq$ not to be entailed. Thus, $\arbtour{}{}{\epred}$ does imply $\loopq$ in the finite setting, and any \bdd rule set for which this implication does not hold (in the unrestricted setting) would disprove \fcbdd.

Note however that \cref{ex-not-fc} does not constitute such a counterexample, as its rule set is simply not \bdd, since the transitivity rule has to be applied at least as many times as the distance between $a$ and $b$ to entail $\epred(a,b)$, which cannot be bounded independently of the size of the database.
Expanding on that example, one could try to mimic the behavior using \bdd rules by replacing the transitivity rule with $\forall x \forall x' \forall y \forall y'~ \epred(x,x') \wedge \epred(y,y') \rightarrow \epred(x,y')$, which entails it. With this change, the rule set is finally \bdd, and the chase entails $\arbtour{}{}{\epred}$. However, this new rule triggers the entailment of $\exists x~ \epred(x,x)$ as soon as $\exists x \exists y~ \epred(x,y)$ is entailed, as expected due to Property~(\sympawn).

While \cref{thm:main} does not entail the \fcbdd conjecture, as a different type of counterexample could in principle exist, we are confident that the insights and carefully curated toolkit presented in the following sections represent an important milestone in the effort to settle the \fcbdd conjecture in the general setting.

\paragraph{Organization of the paper}
We first provide the necessary preliminaries, before stating our main result in Section~\ref{sec:main-result}. The proof is performed in two main steps. 
The first step (\cref{sec:surgeries}) involves a series of rule set ``surgeries''. While some of these are well-known, the properties we require are tailored to the specific case at hand and need a separate formal proof. The second step (\cref{sec:proof-core}) primarily consists of a complex model-theoretic argument, which demonstrates that Property~(\sympawn) holds for the transformed rules.


\section{Preliminaries}

We assume that the reader is familiar with first-order logic and basic concepts from graph theory.
We only provide a brief recap of these topics in this section.

\subsection{First-Order Logic}
We define \textsf{Preds}, and \textsf{Vars} to be mutually disjoint and countably infinite sets of predicates, and variables, respectively.
We associate every $\ppred \in \textsf{Preds}$ with some \emph{arity} $\arity{\ppred} \geq 0$.
For every $n \geq 0$, the set of all $n$-ary predicates is also countably infinite.
A \emph{signature} is a set of predicates.
We write tuples $(x_1, \ldots, x_n)$ of variables as $\vx$ and often treat these tuples as sets.
The \emph{support} of a tuple of variables is the set of variables that occur in it.
A tuple of variables $\vy=(y_1,\dots,y_m)$ is \emph{compatible} with another tuple $\vx=(x_1,\dots,x_n)$ if $m=n$, and for all $i,j\leq n$, we have $y_i=y_j$ whenever $x_i=x_j$.
If, furthermore, for all $i$, either $y_i=x_i$ or $y_i=y_j=x_j$ for some $j$, then $\vy$ is a \emph{specialization} of $\vx$.
Any two tuples $\vx=(x_1,\dots,x_n)$ and $\vt=(t_1,\dots,t_n)$ such that $\vt$ is compatible with $\vx$ define a substitution $[\vx\mapsto\vt]:x_i\mapsto t_i$.

An \emph{atom} over a predicate $\predicate{P}$, or $\predicate{P}$-atom, is a first-order formula of the form $\predicate{P}(\vx)$ with $\vx$ a tuple of variables and $\predicate{P}$ a $\vert \vx \vert$-ary predicate.
For a nullary predicate $\predicate{P}$, we write $\predicate{P}$ as a shortcut for the nullary atom $\predicate{P}()$.
An \emph{instance} over signature $\sig$ is a set of atoms over predicates in $\sig$. For simplicity of presentation, we assume that all instances contain a nullary fact $\top$.
The \emph{active domain} of an atom or an instance is the set of variables that occur in it. The active domain of an atom (or set thereof) $S$ is denoted with $\adom{S}$.
The \emph{disjoint union} $\instance_1\disjunion\instance_2$ of two instances $\instance_1$ and $\instance_2$ is $\instance_1\cup\sigma(\instance_2)$, where $\sigma$ is a function mapping variables of $\instance_2$ to fresh variables, that do not occur in $\instance_1$.
For a first-order formula $F$ and a tuple $\vx$ of variables, we write $F(\vx)$ to indicate that $\vx$ is the set of all free variables occurring in $F$.

\paragraph{Existential rules}
An \emph{(existential) rule} $\eru$ is an FO formula of the form: $$\forall \vx,\vy~ B(\vx, \vy) \rightarrow \exists \vz~ H(\vec{y}, \vec{z})$$ where $B$ and $H$ are non-empty atom conjunctions.
We refer to $\vec{y}$ as the \emph{frontier} of $\eru$, and to $B(\vec{x}, \vy)$ and $\exists \vec{z}~ H(\vec{y}, \vec{z})$ as the \emph{body} and the \emph{head} of $\eru$, respectively.
Such a rule is \emph{Datalog} if it does not feature existential variables; that is, if $\vec{z}$ is empty.
Given a rule $\eru$, the sets $\body{\eru}$, $\head{\eru}$ and $\fr{\eru}$ denote the body, the head and the frontier of $\eru$, respectively.
We refer to elements of $\body{\eru}$ (resp. $\head{\eru}$) as \emph{body-atoms} (resp. \emph{head-atoms}) of $\eru$.
We omit $\forall$ quantifiers when writing rules.
A rule set is over a signature $\sig$ if all the atoms occuring in any of its rules are over a predicate in $\sig$.

\paragraph{Homomorphisms}
A \emph{substitution} $\pi$ is a function from \textsf{Vars} to \textsf{Vars}.
For an atom $\predicate{P}(\vec{t})$, let $\pi(\predicate{P}(\vec{t}))$ be the atom that results from replacing all occurrences of every variable $x$ in $\predicate{P}(\vec{t})$ with $\pi(x)$ if the latter is defined.
A \emph{homomorphism} $\pi$ from an instance $\mathcal{A}$ to an instance $\mathcal{B}$ is a substitution such that $\pi(\mathcal{A}) \subseteq \mathcal{B}$.
An \emph{isomorphism} $\pi$ from $\mathcal{A}$ to $\mathcal{B}$ is an injective and surjective homomorphism from $\mathcal{A}$ to $\mathcal{B}$. 
Two atom sets $\mathcal{A}$ and $\mathcal{B}$ are \emph{homomorphically equivalent} if there is a homomorphism from $\mathcal{A}$ to $\mathcal{B}$ and a homomorphism from $\mathcal{B}$ to $\mathcal{A}$. We denote homomorphic equivalence with $\homequiv$ symbol.

\paragraph{Queries}
A \emph{conjunctive query} (CQ) $\cq(\vx)$ is a pair of a first-order formula of the form $\exists \vz~ B(\vx,\vz)$, where $B$ is a non-empty conjunction of atoms, and a tuple $\vx$ of \emph{answer variables} whose support is the set of free variables of the formula.
A CQ is \emph{Boolean} if it features no answer variable.
A \emph{union of conjunctive queries} (UCQ) $\ucq(\vx)$ is a pair of disjunction of CQs, that we often see as a set of CQs $\set{\cq_1(\vx_1),\dots,\cq_n(\vx_n)}$, and a tuple of answer variables $\vx$ such that for all $i$, $\vx_i$ is a specialization $\vx$.
Given a UCQ $\ucq(\vx,\vy)$ and a tuple of variables $\vx'$ compatible with $\vx$, we denote with $\ucq(\vx',\vy)$ the UCQ $\ucq[\vx\mapsto\vx'](\vx',\vy)$.
Under standard first-order semantics, an instance $\instance$ entails a UCQ $\ucq(\vx)=\set{\cq_1(\vx_1),\dots,\cq_n(\vx_n)}$ for a given tuple $\vt$ of elements of $\adom{I}$ compatible with $\vx$ if there is a homomorphism from $\cq_k(\vx_k[\vx\mapsto\vt])$ to $\instance$ for some $1 \leq k \leq n$. If, furthermore, this homomorphism is injective, then $\instance$ \emph{injectively entails} $\ucq$. We denote these facts with $\instance\models \ucq(\vt)$ and $\instance\modelsinj \ucq(\vt)$, resp.
For a rule set $\rs$, an instance $\inst$ and a UCQ $\ucq$, we write $\pair{\rs, \instance} \models \ucq(\vt)$ to indicate that $\rs \cup \instance$ entails (over unrestricted models) $\ucq(\vx)$ for $\vt$, and $\pair{\rs, \instance} \modelsinj \ucq(\vt)$ when $\rs \cup \instance$ injectively entails (over unrestricted models) $\ucq(\vx)$ for $\vt$.

\subsection{The Chase}
Given a rule set $\rs$, an ($\rs$-)\emph{trigger} over some instance $\instance$ is a pair $\fulltrig$ of a rule $\eru\in\rs$ and a homomorphism $\hom$ from $\body{\eru}$ to some subset of $\instance$.
The \emph{output} of a trigger $\trig=\fulltrig$ over $\instance$ is the instance $\trigoutput{\trig}=\hom'(\head{\eru})$, where $\hom'$ is a homomorphism extending $\hom$ by mapping existentially quantified variables of $\head{\eru}$ to fresh variables. With $\function{triggers}(\inst, \rs)$ we denote the set of all triggers for $\inst$ and $\rs$.
Step $k$ of the (oblivious) chase~\cite{oblivious_chase_original,oblivious_chase_modern} is denoted with $\step{k}{\instance,\rs}$, and defined by
\[\small\step{0}{\instance,\rs}=\instance\quad\quad\step{n+1}{\instance,\rs}\;=\;\step{n}{\instance,\rs}\cup\bigcup_{\trig\in T_n}\trigoutput{\trig} \tag{for $n \in \mathbb{N}_{>0}$}\]
where $T_n=\function{triggers}(\step{n}{\instance,\rs}, \rs)\setminus\function{triggers}(\step{n-1}{\instance,\rs}, \rs)$.
The \emph{result of the chase} from $\pair{\instance,\rs}$, denoted with $\chase{\instance,\rs}$, is the union of $\step{n}{\instance,\rs}$ over all $n$. It is well known that $\chase{\instance,\rs}$ is a universal model of $\instance$ and $\rs$. To ease the reading, we denote $\step{n}{\topinst,\rs}$ as $\step{n}{\rs}$ and $\chase{\topinst,\rs}$ as $\chase{\rs}$.

Given an instance $\instance$ and a rule set $\rs$, a \emph{chase term} is an element of $\adom{\chase{\instance,\rs}}\setminus\adom{\instance}$. Thus, every chase term is created by applying some trigger.
The \emph{frontier} of a chase term $t$ created by some trigger $\fulltrig$ is the set $\hom(\fr{\eru})$.

\subsection{UCQ-rewritability}

\begin{definition}\label{def:ucq-rewritability}
    A rule set $\rs$ is \emph{UCQ-rewritable} if for all CQs $\cq(\vx)$, there is a UCQ $\ucq(\vx)$ such that for all instances $\instance$ and tuples $\vt$ of elements $\adom{\instance}$ compatible with $\vx$, we have
    \[\pair{\instance,\rs}\models \cq(\vt) \iff \instance \models \ucq(\vt).\]
    The UCQ $\ucq(\vx)$ is a (UCQ-)\emph{rewriting} of $\cq(\vx)$ and is denoted with $\rew(\cq(\vx),\rs)$.
\end{definition}

Note that $\rew(\cq(\vx),\rs)$ is not unique; however, we can take the minimal one, which is uniquely defined up to bijective renaming of variables \cite{sound_complete_minimal_ucq_rewritings}. Moreover, the notion of UCQ rewriting can easily be lifted from CQs to UCQs.

\begin{definition}
    A rule set has \emph{bounded derivation depth (bdd)} if for every CQ $\cq$, there is some $k\geq 0$ such that for all instances $\instance$, we have $\pair{\instance,\rs}\models \cq$ if and only if  $\step{k}{\instance,\rs}\models \cq$. We use \bdd to denote the class of rule sets with bdd. For a given $\cq$, we call the minimal $k$ as above the \bdd-\emph{constant} of $\cq$ and denote it with $\bdd(\cq,\rs)$.
\end{definition}

Interestingly the notions coincide. For a full proof of this result, see \cite{bdd_is_fus}.

\begin{proposition}\label{prop:ucq-rewritable_is_bdd}
    A rule set is UCQ-rewritable if and only if it is \bdd.
\end{proposition}

\begin{lemma}\label{lem:bdd-stacking}
    Let $\rs_1$ and $\rs_2$ be \bdd. If for every instance $\inst$ it holds that $\chase{\chase{\inst, \rs_1}, \rs_2}$ is homomorphically equivalent to $\chase{\inst, \rs_1 \cup \rs_2}$ then $\rs_1 \cup \rs_2$ is \bdd.
\end{lemma}
\begin{proof}
    Let $\ucq$ be a rewriting of a CQ $q$ against $\rs_2$ and then $\ucq'$ be a rewriting of $\ucq$ against $\rs_1$. Note that $\ucq'$ is a rewriting of $\cq$ against $\rs_1 \cup \rs_2$ as $\chase{\chase{\inst, \rs_1}, \rs_2} \homequiv \chase{\inst, \rs_1 \cup \rs_2}$.
\end{proof}

Later, we will need injective rewritings, meaning we would like to have \cref{def:ucq-rewritability} rephrased:

\begin{customdef}{\ref{def:ucq-rewritability}}[rephrased]
    A rule set $\rs$ is \emph{UCQ-rewritable} if for all CQs $\cq(\vx)$, there is a UCQ $\ucq(\vx)$ such that for all instances $\instance$ and tuples $\vt$ of elements $\adom{\instance}$ compatible with $\vx$, we have:
    \[\pair{\instance,\rs}\models \cq(\vt) \iff \instance \modelsinj \ucq(\vt)\]
    The UCQ $\ucq(\vx)$ is an 
    \emph{injective rewriting} of $\cq(\vx)$ and is denoted with $\rew_{\mathsf{inj}}(\cq(\vx),\rs)$.
\end{customdef}

The above is indeed an equivalent definition due to the following:

\begin{proposition}\label{prop:injective-rewritings}
    For every UCQ $\ucq$ there exists another UCQ $\ucq_{\mathsf{inj}}$ s.t.
    for every instance $\inst$ and tuple $\va$:
    $$\inst \models \ucq(\va) \;\;\iff\;\; \exists{\cq \in \ucq_{\mathsf{inj}}}\quad \inst \models_{\mathsf{inj}} \cq(\va) \;\;\iff\;\; \inst \models \ucq_{\mathsf{inj}}(\va).$$
\end{proposition}
\begin{proof}
    We will show the case when $Q$ is a CQ, the proof easily generalizes to the UCQ case. Let $\cq$ be a CQ and $\vx$ be the tuple of its variables. Note that whenever $\cq$ maps to some instance $\inst$ via a homomorphism $h$ we have that there exists a specialization $\vy$ of $\vx$ and an injective homomorphism $h'$ such that $[\vx \mapsto \vy] \circ h' = h$. Then the required UCQ is simply a disjunction of $\cq[\vx \mapsto \vy]$ for all specializations $\vy$ of $\vx$. To get the second equivalence note the above construction is idempotent. 
\end{proof}

\subsection{Other Useful Notions}

\paragraph{Graph Theoretical Notions}

A \emph{(directed) graph} is a pair $\pair{V,E}$ composed of a set $V$ of \emph{vertices}, and a subset $E$ of $V\times V$ whose elements are called \emph{edges}.
The \emph{size} of a graph is the number of its vertices.
We often say that an edge $\spair{v,w} \in E$ goes \emph{from} $v$ \emph{to} $w$.
A \emph{loop} is an edge from a node to itself.
A \emph{subgraph} $\pair{V',E'}$ of a graph $\pair{V,E}$ is a graph such that $V'\subseteq V$ and $E\subseteq E'$.
A subgraph $\pair{V',E'}$ of $\pair{V,E}$ is \emph{induced} if for all $v,w\in V'$, if $\spair{v,w}\in E$, then $\spair{v,w}\in E'$.
A \emph{path} between two vertices $s$ and $t$ is a sequence of vertices $v_0,\dots, v_n$ such that $v_0=s$, $v_n=t$ and for all $i<n$, $\spair{v_i,v_{i+1}}\in E$.
A \emph{cycle} is a path from a node to itself.
A graph is \emph{acyclic} if it does not contain any cycle.
A \emph{tournament} is a graph such that for all distinct vertices $v,w$, we have $\spair{v,w}\in E$ or\footnote{Note that this ``or'' is inclusive, contrary to the usual definition of a tournament.} $\spair{w,v}\in E$.
A \emph{$k$-coloring} of the edges of a graph $\pair{V,E}$ is a function from $E$ to $\set{1,\dots,k}$.\footnote{Note that we do not impose adjacent vertices to have different colors.}
Given a graph $G$ and a $k$-coloring $C$ of its edges, a subgraph $\pair{V',E'}$ of $G$ is \emph{colored with $i$} for some $i\leq k$ if for all $e\in E$, we have $C(e)=i$.
The following is a rephrasing of the Ramsey's theorem for directed graphs.
\begin{theorem}\label{cor:Ramsey}
    For any integers $s_1,\ldots,s_k \geq 1$, there is some integer $R(s_1,\ldots,s_k)$ such that for any tournament $T$ of size at least $R(s_1,\dots,s_k)$ and any $k$-coloring of the edges of $T$, there is a subgraph of $T$ that is a tournament of size $s_i$ colored with some $i$.
\end{theorem}

\paragraph{Multisets}

A \emph{multiset over a set $D$} is a function $M$ from $D$ to $\nats$.
A multiset is \emph{finite} if $\set{x\in D\mid M(x)>0}$ is finite.
The size of a finite multiset $M$ is $|M|=\sum_{x\in D} M(x)$. The empty multiset is $\memptyset:x\mapsto 0$.
Given a finite list $[x_1,\dots,x_n]$ of elements of $D$, we denote the multiset $x\mapsto|\set{i\mid x_i=x}|$ with $\multiset{x_1,\dots,x_n}$.
Given a function $f$ from some finite set $E$ to $D$, we denote the multiset $x\mapsto|f^{-1}(x)|$ with $\multiset{f(x) \mid x\in E}$.
The \emph{union} of two multisets $M$ and $N$ is the multiset $M\mcup N:x\mapsto M(x)+N(x)$, their intersection is $M\mcap N:x\mapsto \min(M(x),N(x))$, and their difference is $M\msetminus N:x\mapsto \max(M(x)-N(x),0)$.

Consider a strict order $<$ over $D$. The maximum of a multiset $M$ over $D$ is $\mmax(M)=\max\set{x\in D\mid M(x)>0}$ whenever it is defined, and undefined otherwise.
The strict lexicographical order $\slexorder$ over multisets is defined inductively as follows: for all non-empty multisets $M$, $\memptyset\slexorder M$, and given two non-empty multisets $M$ and $N$, we have $M\slexorder N$ if and only if $\mmax(M)<\mmax(N)$, or if $\mmax(M)=\mmax(N)$ and $(M\msetminus\multiset{\mmax(M)}) \slexorder (N\msetminus\multiset{\mmax(N)})$.
The lexicographical order $\lexorder$ is defined by $M\lexorder N$ if and only if $M\slexorder N$ or $M=N$ as usual.

\begin{lemma}\label{lem:lexorder_wf}
    If $<$ is well-founded on $D$, then for all $k$, the order $\slexorder$ is well-founded on the set of multisets over $D$ of size at most $k$.
\end{lemma}
\begin{proof}
    We proceed by induction on $k$. As the only multiset of size $0$ is $\memptyset$, the case $k=0$ is trivial. Then, consider a set $\mathbb{S}_k$ of multisets of size at most $k$ over $D$. Then, since all the elements of $\mathbb{S}_k$ are finite, they all admit a maximum. Since $<$ is well-founded over $D$, the set $\set{\mmax(M)\mid M\in\mathbb{S}_k}$ admits a minimum $m$. Consider the set $\mathbb{S}_{k-1}=\set{M\mid M\mcup\multiset{m}\in\mathbb{S}_k\wedge \mmax(M\mcup\multiset{m})=m}$. As all the multisets in this set have size at most $k-1$, by induction hypothesis, it admits a minimum $M'_{\min}$. Then, we claim that $M_{\min}=M'_{\min}\mcup\multiset{m}$ is the minimum of $\mathbb{S}_k$. Consider an element $M$ of $\mathbb{S}_k$. If $M\msetminus\multiset{m}\notin\mathbb{S}_{k-1}$, then $\mmax(M)>m=\mmax(M_{\min})$ by definition of $\mathbb{S}_{k-1}$, so $M_{\min}\lexorder M$. Otherwise, if $M\msetminus\multiset{m}\in\mathbb{S}_{k-1}$, then $\mmax(M)=m=\mmax(M_{\min})$, and $M_{\min}\msetminus\multiset{m}=M'_{\min} \lexorder M\msetminus\multiset{m}$ as $M'_{min}$ is the minimum of $\mathbb{S}_{k-1}$. Thus, $M_{\min}\lexorder M$, which concludes the proof.
\end{proof}

\section{The Main Result}
\label{sec:main-result}

\begin{definition}\label{def:arb-tournament}
Given a rule set $\rs$, an instance $\instance$ and a binary predicate $\epred$, we write $\arbtour{\rs}{\instance}{\epred}$ to denote the following query:\footnote{In order for $\models$ to make sense with $\arbtour{}{}{\epred}$ we can define the query via the following second order logic formula: $\Forall{\xpred}\text{finite, }\Exists{\ypred}\; \Forall{x,y}\Exists{x',y'}\; \xpred(x) \land \xpred(y) \;\rightarrow\; \ypred(x,x') \land \ypred(y,y') \land \epred(x',y').$ The finiteness of $X$ is expressed by stating that all functions from $X$ to $X$ are injective if and only if they are surjective.} \emph{for all integers $k$, there exists a $\epred$-tournament of size $k$ in the given instance.}
\end{definition}

\begin{definition}\label{def:loop-query}
    We define the $\epred$-\emph{loop query} $\exists{x}~\epred(x,x)$ and denote it with $\loopq$.
\end{definition}

For the rest of the paper, let us fix a binary predicate $\epred$ for use in both $\arbtoure$ and $\loopq$.

\begin{customthm}{\ref{thm:main}}[restated]
For every signature $\signature$, every UCQ-rewritable rule set $\rs$ over $\signature$, and every instance $\instance$ we have: 
\begin{align*}
\mainstatement{\inst}{\rs}.\tag{\textup{\sympawn}}
\end{align*}
\end{customthm}

In the context of this theorem and statements with a similar structure we call (\ref{tag:thm:main}) the \emph{consequence} while the preceding part is the \emph{assumption}. We shall devote the rest of the paper to the proof of the above, main result. 

\section{Reducing the Main Theorem}
\label{sec:surgeries}
The goal of this section is to reduce the assumption of the main theorem of the paper to a more controlled one in which the considered instance is always $\topinst$ and the rule set satisfies a number of properties. This overarching goal is reflected in \cref{thm:main-regal} which - as we show in this section - is equivalent to \cref{thm:main}. The challenge here lies in showing that \cref{thm:main-regal} implies \cref{thm:main} as the other direction is trivial. The proof goes via a number of steps - each showing equivalence to some intermediate statement (\cref{lem:main-a,lem:main-b,lem:main-cde}) with progressively stronger assumptions. The toolkit of the section consists of a number of rule set surgeries of varying degrees of complexity.

All the parts of the proof (of equivalence of \cref{thm:main} and \cref{thm:main-regal}), will follow the same formula, which we will discuss alongside the next step.

\subsection{Encoding Instances in Rule Sets} 
We begin with the first intermediate statement, enforcing the instance in \cref{thm:main} to be $\topinst$.

\begin{lemma}\label{lem:main-a}
    For every signature $\signature$, and every UCQ-rewritable rule set $\rs$ over $\signature$ we have: 
    \begin{align*}
        \mainstatementtopinstance{\topinst}{\rs}.\tag{\textup{\symknight}}\label{tag:lem:main-a}
    \end{align*}
\end{lemma}

The goal of this step is to show that \cref{lem:main-a} is equivalent to \cref{thm:main}.
As we are restricting the set of potential instances, the $(\Rightarrow)$ direction of the proof is trivial, that is \cref{thm:main} implies \cref{lem:main-a}.

In this and the following steps of the proof, we shall always prove the $(\Leftarrow)$ direction via contraposition. Therefore, assume signature $\signature$, rule set $\rs$, and instance $\inst$ form a counterexample to \cref{thm:main}, now we shall construct $\rs'$ such that $\rs'$ and $\signature$ constitute a counterexample to \cref{lem:main-a}.

\subsubsection*{Rule Set Surgery}

The counterexample construction will always consist of some rule set surgery.

\begin{definition}
    For any instance $\jnst$ we denote with $\top \to \jnst$ the following rule \[\narule{\top}{f(\adom{\jnst})}{\bigwedge_{\apred(\vt) \in \jnst} \apred(f(\vt))}\] where $f$ is a bijective renaming of terms to fresh variables.
\end{definition}


\begin{observation}\label{obs:simple-top-to-jnst}
    Let $\eru$ be any rule with $\top$ being its body, then for any rule set $\mathcal{S}$ and any instance $\jnst$ we have:
    $$\chase{\jnst, \mathcal{S} \cup \singleton{\eru}} \homequiv \chase{\chase{\jnst, \singleton{\eru}},\; \mathcal{S}}.$$
\end{observation}
\begin{proof}
    Note that $\eru$ always triggers, and triggers only once.
\end{proof}

We simply let $\rs' = \rs \cup \singleton{\top \to \inst}$. With $\rs'$ constructed we shall show now that the newly constructed rule set $\rs'$ and the original $\signature$ form a counterexample to \cref{lem:main-a}. This will always consist of two parts: showing that the assumption holds - in this case that $\rs'$ is UCQ-rewritable - and that the consequence does not hold - which in this case amounts to showing:
$$\chase{\rs'} \models \arbtoure \;\land\; \chase{\rs'} \not\models \loopq$$

\subsubsection*{Main Properties of the Surgery}
As in the following steps we present here the main properties of the surgery. The raison d’être of such subsections is to provide a modular and reusable formulation of the main building blocks of the proof. As noted in the introduction, some of these techniques are well-known in the community. However, they often lack proper formulation or formal statements. Importantly, the properties of most of the techniques with respect to UCQ-rewritability are far from being well-known and are themselves results of separate interest to the community. Finally, we note that we will usually state two important properties - one used for the ``Falsifying the Consequence'' parts and one for the ``Proving the Assumption'' ones. 
\begin{observation}\label{obs:remove-db-A}
    For every pair of instances $\jnst$ and $\jnst'$ and a rule set $\rsb$ we have that $\chase{\jnst \disjunion \jnst', \rsb}$ is homomorphically equivalent to $\chase{\jnst, \rsb \cup \singleton{\top \to \jnst'}}$.
\end{observation}

\begin{proof} Consider the following sequence of equivalences:
    \begin{align*}
        \chase{\jnst \disjunion \jnst', \rsb} &\homequiv \chase{\;\chase{\jnst, \set{\top \to \jnst'}},\; \rsb\;} \tag{Definition of $\set{\top \to \jnst'}$}\\
        &\homequiv \chase{\jnst, \rsb \cup \set{\top \to \jnst'}} \tag*{(From \cref{obs:simple-top-to-jnst}) \quad \qedhere}
    \end{align*}
\end{proof}

\noindent
The following allows us to encapsulate the first property of the surgery:
\begin{corollary}\label{cor:db-into-rs-prop-one}
    For any instance $\jnst$ and any rule set $\rsb$ we have that:
    $$\chase{\jnst, \rsb} \;\leftrightarrow\; \chase{\topinst, \rsb \cup \singleton{\top \to \jnst}}.$$
\end{corollary}
and the following is its second property:
\begin{observation}\label{cor:db-into-rs-prop-two}
    Given a rule set $\rsb$ that is UCQ-rewritable and an instance $\jnst$ we have that\\ $\rsb \cup \singleton{\top \to \jnst}$ is UCQ-rewritable.
\end{observation}
\begin{proof}
    Using \cref{lem:bdd-stacking,obs:simple-top-to-jnst} we get that $\rsb \cup \singleton{\top \to \jnst}$ is UCQ-rewritable.
\end{proof}

\subsubsection*{Falsifying the Consequence: (\ref{tag:lem:main-a})}
\begin{lemma}
    $\chase{\rs'} \models \arbtoure \;\land\; \chase{\rs'} \not\models \loopq$.
\end{lemma}
\begin{proof}
    Follows from:
\begin{itemize}
    \item  $\chase{\inst, \rs} \homequiv \chase{\rs'}$ \quad(\cref{cor:db-into-rs-prop-one})
    \item $\chase{\inst, \rs} \models \arbtoure \land \chase{\inst, \rs} \not\models \loopq$ \quad(from assumption).\qedhere
\end{itemize}

\end{proof}
\subsubsection*{Proving the Assumption of \cref{lem:main-a}} In this case we have to show that $\rs'$ is UCQ-rewritable. It follows, from  \cref{cor:db-into-rs-prop-two}, as $\rs$ is UCQ-rewritable.

\subsection{Reducing to Binary Signatures}
\begin{lemma}\label{lem:main-b}
    For every binary signature $\signature$, and every UCQ-rewritable rule set $\rs$ over $\signature$ we have: 
    \begin{align*}
    \mainstatementtopinstance{\topinst}{\rs}\tag{\textup{\symbishop}}\label{tag:lem:main-b}
    \end{align*}
\end{lemma}

Now we show that \cref{lem:main-a} is equivalent to \cref{lem:main-b}. The $(\Rightarrow)$ direction is trivial, we show $(\Leftarrow)$ via the contraposition. Let $\signature$ be a signature and $\rs$ be a rule set that together form a counterexample to \cref{lem:main-a}.

\subsubsection*{Reification}
Given a relational symbol $\apred$ of arity $\arity{\apred} > 2$ we define $\reify(\apred)$ as the following set $\set{\apred_1,\ldots,\apred_{\arity{\apred}}}$ of binary predicates.
Given a finite signature $\signature = \signature_{\leq2} \uplus \signature_{\geq3}$ divided into at-most-binary and higher-arity predicates, we define the \emph{reified version} of $\signature$ as the binary signature $\reify(\signature) = \signature_{\leq2} \uplus \bigcup_{\apred \in \signature_{\geq3}} \reify(\apred).$
Given an atom $\alpha = \apred(x_1, \ldots, x_n)$ we define $\reify(\alpha)$ as $\set{\apred_i(x_i, x_\alpha) \mid 1 < i \leq n}$ for $n > 2$ and let $\reify$ be the identity on atoms of arity less than three.
We lift $\reify$ to instances, rules, and queries in the natural way.

We prove that $\rs' = \reify(\rs)$ and $\signature' = \reify(\signature)$ form a counterexample to \cref{lem:main-b}.

\subsubsection*{Main Properties of Reification}
We take our first property from Feller et al.~\cite{cliquewidth-paper} (Property (ii) of reification):
\begin{lemma}\label{lem:reification-prop-one}
    For any instance $\jnst$ and any rule set $\rsb$ we have that $$\chase{\reify(\jnst), \reify(\rsb)} \homequiv \reify(\chase{\jnst, \rsb}).$$
\end{lemma}

The relation of reification to UCQ-rewritability was to our knowledge not explored before:
\begin{lemma}\label{lem:reification-property-two}
    For any UCQ-rewritable rule set $\rsb$ we have that $\reify(\rsb)$ is UCQ-rewritable as well.
\end{lemma}
\begin{proof}
    Given a symbol $\apred$ of signature $\Xi$ of $\rsb$ whose arity is greater than two, let $\eru_{\apred}$ be a rule $$\narule{\apred(x_1, \ldots, x_n)}{z}{\bigwedge_{\apred_i \in \reify(\apred)} \apred_i(x_i, z)}$$ that is $\eru_\apred$ is a rule that simply \emph{projects} $\apred$ to its reified variant. Let $\rsbb$ be $\rsb$ enriched with rules $\eru_\apred$ for all relevant $\apred$ of signature of $\rsb$. Note that $\rsbb$ is UCQ-rewritable - as $\rsb$ is UCQ-rewritable and the added rules cannot be fired recursively or cannot trigger rules from $\rsb$. 
    
    The restriction of $\chase{\jnst, \rsbb}$ to reified signature $\reify(\Xi)$ is isomorphic to $\reify(\chase{\jnst, \rs})$ for any instance $\jnst$. By \cref{lem:reification-prop-one}, it is also isomorphic to $\chase{\reify(\jnst), \reify(\rs)}$. From this we conclude that if we take any CQ $Q$ over $\reify(\Xi)$ and let $\Phi$ be its rewriting against $\rsbb$, we have that $\reify(\Phi)$ is the rewriting of $Q$ against $\reify(\rsb)$ - which implies that $\reify(\rsb)$ is UCQ-rewritable.
\end{proof}

\subsubsection*{Falsifying the Consequence: (\ref{tag:lem:main-b})}
Using \cref{lem:reification-prop-one} we have that $\reify(\chase{\rs})$ is homomorphically equivalent to $\chase{\rs'}$; hence if the consequence (\ref{tag:lem:main-a}) does not hold, then (\ref{tag:lem:main-b}) does not hold as well.

\subsubsection*{Proving the Assumption of \cref{lem:main-b}}
As $\signature'$ is binary and we have \cref{lem:reification-property-two} we have shown that assumption of \cref{lem:main-b} holds for $\rs'$ and signature $\signature'$ which ends the proof.

\subsection{Streamlining the Heads}

We now introduce two definitions that allow us to have a convenient structure in the chase as long as only non-Datalog rules have been applied: atoms are endowed with an orientation, and each term has a unique predecessor by atoms of a given predicate. 

\begin{definition}[Forward-existential]\label{def:rs-forward-existential}
    A rule $\eru$ is \emph{forward-existential} \iffi for each its head-atoms $A(x,y)$ we have that $x$ is a frontier variable and $y$ is an existential variable. A rule set is forward-existential \iffi each of its non-Datalog rule is forward-existential.
\end{definition}

\begin{definition}[Predicate-unique]\label{def:frontier-functional}
    A rule set $\rs$ over binary signature $\signature$ is \emph{predicate-unique} \iffi for every non-Datalog rule $\eru$ of $\rs$ every predicate $\epred \in \signature$ appears at most once in the head of $\eru$.
\end{definition}

Note, that the above is not equivalent with single-head rule sets as \[\narule{\apred(x),\bpred(y)}{z}{\dpred(x,z),\epred(y,z)}\] is a predicate-unique, forward-existential rule.

\begin{lemma}\label{lem:main-cde}
    For every binary signature $\signature$ and every UCQ-rewritable, forward-existential, and predicate-unique rule set $\rs$ over $\signature$ we have: 
    \begin{align*}
        \mainstatementtopinstance{\topinst}{\rs}. \tag{\textup{\symrook}}\label{tag:lem:main-cde}
    \end{align*}
\end{lemma}

Now we show that \cref{lem:main-b} is equivalent to \cref{lem:main-cde}. The $(\Rightarrow)$ direction is trivial, we show $(\Leftarrow)$ via the contraposition. Let $\signature$ be a signature and $\rs$ be a rule set that together form a counterexample to \cref{lem:main-b}.

\subsubsection*{Rule Set Surgery}
\newcommand{\eruinit}{\eru_\mathrm{init}}
\newcommand{\eruDL}{\eru_\mathrm{DL}}
\newcommand{\eruex}{\eru_\exists}
Let $\apred$ and $\bpred$ (with indices) be fresh relational symbols. Given a rule\\ \[\eru = \narule{B(\vx, \vy)}{\vz}{H(\vy, \vz)}\] we define:
\begin{itemize}
    \item the existential rule $\eruinit$ as
    \[\narule{B(\vx, \vy)}{w}{\apred_{0}^\eru(w)\wedge\bigwedge_{y\in\vy} \apred^\eru_{y,w}(y,w)};\]
    \item the existential rule $\eruex$ as
    \[\narule{\apred_0^\eru(w)\wedge\bigwedge_{y\in\vy} \apred_{y,w}^\eru(y,w)}{\vz}{\bigwedge_{y'\in\vy\,\cup \singleton{w}}\bigwedge_{z\in\vz} \bpred^\eru_{y',z}(y',z)};\]
    \item the Datalog rule $\eruDL$ as
    \[\nadrule{\bigwedge_{y'\in\vy\,\cup \singleton{w}}\bigwedge_{z\in\vz} \bpred^\eru_{y',z}(y',z)}{H(\vy, \vz)}.\]
\end{itemize}

\newcommand{\stream}{\bigtriangledown}
Given a rule set $\rsb$ let $\stream(\rsb)$ be the rule set containing $\eruinit$, $\eruex$ and $\eruDL$ for all rules in $\rsb$. Now let $\rs' = \stream(\rs)$, and let $\signature'$ contain all predicates that appear in $\rs'$. We prove that $\rs'$ and $\signature'$ form a counterexample to \cref{lem:main-cde}.

\subsubsection*{Main Properties of Streamlining}
As the rules of $\stream(\rsb)$ simply introduce intermediate steps to the chase --- note there is no room for unwanted interplay between the rules --- we have the following:

\begin{lemma}\label{lem:streamline-prop-one}
For every rule signature $\sig$, rule set $\rsb$ and instance $\jnst$ over $\sig$, we have that $\chase{\jnst, \rsb}$ is homomorphically equivalent to $\chase{\jnst, \stream(\rsb)}$ when restricted to $\sig$.
\end{lemma}
\begin{proof}
The proof can be found in \cref{app:streamline-prop-one}.
\end{proof}

\begin{lemma}\label{lem:streamline-prop-two}
Given a rule set $\rsb$ we have that $\stream(\rsb)$ is forward-existential and predicate-unique. Moreover, if $\rsb$ is UCQ-rewritable then $\stream(\rsb)$ is UCQ-rewritable as well.
\end{lemma}
\begin{proof}
The proof can be found in  \cref{app:C}
\end{proof}

\subsubsection*{Falsifying the Consequence: (\ref{tag:lem:main-cde})}
From \cref{lem:streamline-prop-one} we have that the chase of $\topinst$ and $\rs$ and the chase of $\topinst$ and $\rs'$ are homomorphically equivalent when restricted to signature $\signature$. Therefore, if the consequence (\ref{tag:lem:main-b}) does not hold then (\ref{tag:lem:main-cde}) does not hold as well.

\subsubsection*{Proving the Assumption of \cref{lem:main-cde}}
Given $\rs$ is UCQ-rewritable we have that $\rs'$ is as required due to \cref{lem:streamline-prop-two}.

\subsection{Rewriting Bodies}

In this section we will use the quickness property from \cite{sebastian-piotr-rpq-paper}, inspired from \cite{DBLP:conf/pods/Ostropolski-Nalewaja22}. 

\begin{definition}[Quick]\label{def:rs-quick}
   A rule set $\rs'$ is {\em quick} \iffi for every instance $\inst$ and every atom $\beta$ of $\chase{\inst, \rs'}$ if all frontier terms of $\beta$ appear in $\adom{\inst}$ then $\beta \in \step{1}{\inst, \rs'}$.
\end{definition}

\begin{definition}[Regal]\label{def:rs-regal}
    A rule set $\rs$ over binary signature is \emph{regal}
    \iffi it is UCQ-rewritable, quick, forward-existential, and predicate-unique.
\end{definition}

We can now state \cref{thm:main-regal}, which is the statement that we prove in Section~\ref{sec:proof-core}.

\begin{theorem}\label{thm:main-regal}
    For every binary signature $\signature$ and every regal rule set $\rs$ over $\signature$ we have: 
    \begin{align*}
        \mainstatementtopinstance{\topinst}{\rs}.\tag{\textup{\largecrown}}\label{tag:thm:main-regal}
    \end{align*}
\end{theorem}

We now show that \cref{lem:main-cde} is equivalent to \cref{thm:main-regal}. The $(\Rightarrow)$ direction is trivial, we show $(\Leftarrow)$ via the contraposition. Let $\signature$ be a signature and $\rs$ be a rule set that together form a counterexample to \cref{lem:main-cde}.
\subsubsection*{Rule Set Surgery}
\newcommand{\rewop}{\mathtt{rew}}
\newcommand{\rrew}[1]{\rew(#1)}
\newcommand{\rhorew}[2]{\rew(#1,#2)}

Following \cite{sebastian-piotr-rpq-paper} consider the next definition:
\begin{definition}[Body Rewriting]\label{def:body-rewriting}
Given a rule set $\rsb$ and an existential rule $\rho \in \rsb$ of the form $\cnarule{B(\vx, \vy)}{\vz}{H(\vy, \vz)}$,
let $\rhorew{\rho}{\rsb}$ be the rule set:
\begin{align*}
    &\Big\{\ \narule{q(\vx', \vy')}{\vz}{H(\vy', \vz) \ \Big| \\[-1ex] 
    &\hspace{5ex}\Exists{\vx'}q(\vx', \vy') \;\in\; \rew(\Exists{\vx} B(\vx, \vy), \rsb)}\ \Big\}.
\end{align*}
Finally, let $\rrew{\rsb}= \rsb \cup \bigcup_{\rho \in \rsb}\; \rhorew{\rho}{\rsb}.$
\end{definition}

Let $\rs'$ be $\rrew{\rs}$. We now show that $\rs'$ and $\signature$ form a counterexample to \cref{thm:main-regal}.
\subsubsection*{Main Properties of Body Rewriting}
We take the following (Lemma 42, \cite{sebastian-piotr-rpq-paper}) as our first property of the surgery:
\begin{lemma}[\cite{sebastian-piotr-rpq-paper}]\label{lem:rew-prop-one}
For any \bdd rule set $\rsb$ and any instance $\jnst$ we have $\chase{\jnst, \rsb} \homequiv \chase{\jnst, \rew(\rsb)}$.
\end{lemma}

As an easy consequence of Lemma~\ref{lem:rew-prop-one}, we have the following. 
\begin{lemma}\label{lem:rew-prop-two}
    The $\rew$ rule set surgery preserves UCQ-rewritable, predicate-unique, and forward-existential properties of rule sets.
\end{lemma}
\begin{proof}
Let $\rsb$ be a rule set that is:
\begin{itemize}
    \item \emph{UCQ-rewritable}:  As chases of $\rsb$ and $\rew(\rsb)$ are homomorphically equivalent (\cref{lem:rew-prop-one}), the rewritings against $\rsb$ can simply serve as rewritings against $\rew(\rsb)$.
    \item \emph{Predicate-unique or forward-existential}: Note that the heads of rules of $\rew(\rsb)$ are unchanged 
    — except potential identification of frontier variables —
    when compared to $\rsb$. Therefore, as $\rsb$ is predicate-unique and forward-existential we have that $\rew(\rsb)$ is as well. \qedhere
\end{itemize}
  
\end{proof}

Moreover, we take the following (Lemma 28, \cite{sebastian-piotr-rpq-paper})
\begin{lemma}[\cite{sebastian-piotr-rpq-paper}]\label{lem:frontier}
For any UCQ-rewritable rule set $\rsb$, $\rew(\rsb)$ is quick.
\end{lemma}

\subsubsection{Main Properties of Regality}
\newcommand{\rsbc}{\rsb}
\newcommand{\rsbce}{\rsb^\exists}
\newcommand{\rsbcdl}{\rsb^{DL}}
\newcommand{\chasersbexists}{\step{\exists}{\topinst, \rsb}}
\newcommand{\chasersbexiststopinstance}{\step{\exists}{\rsb}}

In this section we note all relevant properties of regal rule sets that we use in the following section. Let $\rsbc$ be a regal rule set and $\rsbcdl$ and $\rsbce$ be its subsets of Datalog rules an non-Datalog rules, respectively.
%
We start by noticing that, due to the quickness of $\rsb$, $\chase{\rsb}$ can be obtained by applying only Datalog rules on top of $\chase{\rsbce}$.

\begin{lemma}\label{lem:skeleton}
 $\chase{\rsbc}$ and $\chase{\chase{\rsbce}, \rsbcdl}$ are homomorphically equivalent.
\end{lemma}
\begin{proof}
\newcommand{\chasebc}{\chase{\rsbc}}
\newcommand{\chasebce}{\chase{\rsbce}}
    The existence of a homomorphism from $\chase{\chasebce, \rsbcdl}$ to $\chase{\rsbc}$ is clear, as there is a homomorphism from $\chasebce$ to $\chasebc$ (because $\rsbce \subseteq \rsbc$), and $\rsbcdl \subseteq \rsbc$. Conversely, we prove by induction on $i$ the existence of $\hom_i$ (extending $\hom_{i-1}$ if $i>0$) from the terms of $\step{i}{\rsbc}$ to the terms of $\step{i}{\rsbce}$ such that $\hom_i$ is a homomorphism from $\step{i}{\rsbc}$ to $\chase{\step{i}{\rsbce},\rsbcdl}$. For $i=0$, the result holds trivially. If the result holds for $i \geq 0$, let us consider a term $t$ that belongs to $\step{i+1}{\rsbc}$ but not to $\step{i}{\rsbc}$. It must have been created by a non-Datalog rule $\rho$ whose body maps by $\hom$ to $\step{i}{\rsbc}$. Thus $\hom_i\circ\hom$ is a homomorphism from the body of $\rho$ to $\chase{\step{i}{\rsbce},\rsbcdl}$, and all frontier terms of the created atoms belong to $\step{i}{\rsbce}$. By quickness of $\rsbcdl \cup \rsbce$, these atoms can be created by a single application of a rule $\eru$, which must belong to $\rsbce$, as it is non-Datalog. We extend $\hom_i$ by mapping $t$ to the corresponding term created by that application of $\eru$. Regarding Datalog rule applications, let $\fulltrig$ be a trigger applied at step $i+1$. $\hompi(\body{\eru}) \subseteq \step{i}{\rsbc}$, hence $\langle\eru,\hom_i\circ\hompi\rangle$ is a trigger applicable on $\chase{\step{i}{\rsbce},\rsbcdl}$, and its result thus belongs to $\chase{\step{i}{\rsbce},\rsbcdl}$.
\end{proof}

However, $\chase{\rsbce}$ has a much nicer structure than $\chase{\rsbc}$, that we will exploit shortly. To describe this structure, we associate with each term a timestamp.

\begin{definition}
Given a term $t$ of $\chase{\rsbce}$ we define the \emph{timestamp} $\timestamp{t}$ of $t$ to be the smallest integer $i$ such that $t\in\adom{\step{i}{\rsbce}}$. Given a set of terms $T$ of $\chase{\rsbce}$, we define $\timestamps{T}$ as the multiset $\multiset{\timestamp{t}\mid t\in T}$.
\end{definition}


Given that the signature of $\rsbc$ is binary, we can see instances (and queries) as directed graphs.


\begin{observation}\label{obs:chasece-is-dag}
    $\chase{\rsbce}$ is a directed acyclic graph.
\end{observation}
\begin{proof}
    As $\rsbc$ is forward-existential (and the database does not contain any terms), for all atoms $\apred(t,u)$ in $\chase{\rsbce}$, we have $\timestamp{t}<\timestamp{u}$. Thus, if some term $t$ is part of a cycle, we get that $\timestamp{t}<\timestamp{t}$, which is a contradiction.
\end{proof}

\subsubsection*{Falsifying the Consequence: (\ref{tag:thm:main-regal})}

From \cref{lem:rew-prop-one} we have that the chase of $\topinst$ and $\rs$ and the chase of $\topinst$ and $\rs'$ are homomorphically equivalent. Therefore, if the consequence (\ref{tag:lem:main-b}) does not hold then (\ref{tag:lem:main-cde}) does not hold as well.

\subsubsection*{Proving the Assumption of \cref{thm:main-regal}}

Knowing that $\rew$ preserves UCQ-rewritable, predicate-unique, and forward-existential properties of rule sets (\cref{lem:rew-prop-two}) and that $\rs'$ is quick (\cref{lem:frontier}), we know that $\rs'$ is regal.

\section{Proving \cref{thm:main-regal}{Theorem~\ref{thm:main-regal}}}
\label{sec:proof-core}





Let us fix a regal rule set $\rsc$. Let $\rsce$ define its non-Datalog part and $\rscdl$ denote the set of Datalog rules of $\rsc$. 
We shall prove \cref{thm:main-regal} by contradiction. Assume $\arbtour{\epred}{\epred}{\epred}$ holds in $\chasec$ but $\loopq$ does not. By \cref{lem:skeleton}, $\arbtour{\epred}{\epred}{\epred}$ also holds in $\chase{\chasece,\rscdl}$, but $\loopq$ does not. Let $K_n$, for $n \in \nats$, be an $\epred$-tournament of size $n$ in $\chase{\chasece,\rscdl}$.

The idea behind the proof is as follows. Due to the regality of $\rsc$, for each $\epred(s,t)$ of $K_n$, the injective rewriting $\qc$ of $\epred$ against $\rsc$ holds for $s$ and $t$ in $\chasece$. 
In \cref{sec:proof-core:n-valleys-to-1} we prove that arbitrary large tournaments can be witnessed by a single disjunct of $\qc$ of a specific shape (called a \emph{valley query}). Finally, in \cref{sec:single-valley}, we prove that a single valley query can only generate cliques of size $3$ without creating loops, which will conclude the proof.

\subsection{Valley Queries}\label{sec:proof-core:n-valleys-to-1}


 Let $\qc$ be the injective rewriting of $\epred(x,y)$ against $\rsc$, which exists due to \cref{prop:injective-rewritings} and because $\rsc$ is UCQ-rewritable.

\begin{definition}\label{def:witnesses}
    Given $s$ and $t$ such that $\epred(s,t)\in \chase{\chasece,\rscdl}$, we define the set of \emph{witnesses} of $\epred(s,t)$ in $\chasece$ as $\wn{s,t}=\set{\cq\in \qc\mid \chasece\modelsinj \cq(s,t)}$.
\end{definition}

\begin{observation}\label{obs:wn_nonempty}
    For all $s$ and $t$ such that $\epred(s,t)\in \chase{\chasece,\rscdl}$, the set $\wn{s,t}$ is not empty.
\end{observation}
\begin{proof} 
Follows from \cref{prop:injective-rewritings}, and the fact that $\qc$ is an injective rewriting of $\epred(x,y)$ against $\rsc$, which contains $\rscdl$.
\end{proof}

The acyclicity of $\chasece$ allows us to find queries of a special shape within $\wn{s,t}$.

\newcommand{\sordc}{<_{\crown}}
\newcommand{\ordc}{\leq_{\crown}}

\begin{definition}
    Given an instance $\inst$ (or a CQ) that is a directed acyclic graph, we define a strict partial order $<_{\inst}$ over terms of $\inst$ as follows: given two terms $s, t$ of $\inst$ we have $s <_{\inst} t$ \iffi there exists a directed path from $s$ to $t$. We define $\leq_{\inst}$ as the reflexive closure of $<_{\inst}$. We use $\sordc$ and $\ordc$ as shorthands for $<_{\chasece}$ and $\leq_{\chasece}$ respectively.
\end{definition}

\begin{definition}
    A \emph{valley query} $\cq(x,y)$ is a binary conjunctive query that is a DAG and whose only $<_{\cq}$-maximal variables are $x$ and $y$.
\end{definition}


\begin{lemma}[Peak Removing Argument]
\label{lem:peak}
    For all $s$ and $t$ such that $\epred(s,t) \in \chase{\chasece,\rscdl}$, the set $\wn{s,t}$ contains a valley query.
\end{lemma}
\begin{proof}
    Assume towards contradiction that this is not the case. Consider a non-valley query $\cq\in \wn{s,t}$ and an injective homomorphism $\hom$ from $\cq$ to $\chasece$ such that $\hom(x)=s$, $\hom(y)=t$, and $\timestamps{\hom(\cq)}$ is $\lexorder$-minimal among all such queries and homomorphisms. Such a pair exists by \cref{obs:wn_nonempty} and the fact that $\lexorder$ is well-founded over multisets of size at most $\max\set{|\cq'|\mid \cq'\in\qc}$, by \cref{lem:lexorder_wf}.
    
    Then, consider a variable $z\in\exvars{\cq}$ that is $\leq_{\cq}$-maximal. Such a variable exists as $\cq$ is not a valley query. Since $\hom$ is injective, $\hom(z)$ does not have any outgoing edge in $\hom(\cq)$. Denote the set of atoms containing the variable $z$ in $\cq$ with $Z$, and let $\fulltrigpi$ be the trigger that created $\hom(z)$. Since $\rsc$ is forward-existential and $z$ is $\leq_{\cq}$-maximal, then $\hom(Z)\subseteq\pi(\head{\eru})$. We thus define $\instance=\hom(\cq)\setminus \hom(Z)\cup \hompi(\body{\eru})$. Note that by this definition, $\inst\subseteq\chasece$.
    
    The trigger $\fulltrigpi$ is applicable on $\instance$ and yields $\pi(\head{\eru})$ which contains $\hom(Z)$, up to the renaming of $\hom(z)$. Thus, since $z$ is existentially quantified in $\cq$, we have $\chase{\instance,\rsc}\models \cq(s,t)$. Thus, $\chase{\instance,\rsc}\models \epred(s,t)$, which by the fact that $\qc(x,y)$ is an injective rewriting of $\epred(x,y)$ entails that $\instance\modelsinj \cq'(s,t)$ for some $\cq'\in\qc$. Consider the injective homomorphism $h'$ that witnesses this fact.

    First note that as $h'(\cq')\subseteq\instance$, we have $\timestamps{h'(\cq')}\lexorder\timestamps{\inst}$. Then, since $\hom(z)$ is created by the trigger $\fulltrigpi$, by definition of the chase, we have $\timestamp{x}<\timestamp{\hom(z)}$ for all $x\in\hompi(\body{\eru})$. Thus, $\timestamps{I}\slexorder\timestamps{\hom(\cq)}$, which entails that $\timestamps{h'(\cq')}\slexorder\timestamps{\hom(\cq)}$, contradicting the minimality of $\pair{\cq,h}$. 
\end{proof}

\begin{proposition}
    If there exists $\epred$-tournaments of arbitrary size in $\chasec$, there exists tournaments of arbitrary size defined by a single valley query over $\chasece$.
\end{proposition}

\begin{proof}
    First, as noted earlier, by \cref{lem:skeleton}, if there are $\epred$-tournaments of any size in $\chasec$, there are $\epred$-tournaments of arbitrary size in $\chase{\chasece,\rscdl}$. We then make use of Ramsey's theorem. 

For a given $\epred$-tournament in $\chase{\chasece,\rscdl}$, we color each edge $E(s,t)$ by an arbitrary valley query in $\wn{s,t}$. This possible due to \cref{lem:peak}. Ramsey theorem ensures that for any $n$, if the original tournament is larger than $R(n,\ldots,n)$ (with as many arguments as elements of $\qc$), there is a tournament of size $n$ defined by a single valley query. 
\end{proof}

\subsection{Tournaments Defined by a Single Valley Query}\label{sec:single-valley}

We now analyze the cliques that can be defined by a single valley query. Our main tool here is that the image of a valley query in $\chasece$ is entirely determined by the image of its answer variables. This strong property holds for valley queries because $\rsc$ is both predicate-unique and forward-existential. 


\begin{lemma}\label{lem:cq-functional}
    Let $\cq(x, \vy)$ be a CQ such that $y <_{\cq} x$ for every $y \in \vy$ then the set \[\set{\pair{s, \vt} \in \adom{\chasece} \times \adom{\chasece}^{|\vy|} \quad\mid\quad \chasece \models \cq(s, \vt)}\] is a function.
\end{lemma}
\begin{proof}
    We first show this result in the case where $|\vy|=1$ (and $y$ is the only variable in $\vy$), and $\cq(x,y)$ is a path from $y$ to $x$. We prove by induction on the length of the path from $y$ to $x$ in $\cq(x,y)$ that for all $s\in\adom{\chasece}$, there is at most one $t\in\adom{\chasece}$ such that $\chasece\models\cq(s,t)$. The case where this path has length $0$ being trivial, we assume it has at least size one. Thus, $\cq(x,y)=\cq'(x,z)\wedge\apred(y,z)$, where $\apred$ is a predicate in $\signature$, and $\cq'(x,z)$ is a path query of length $|\cq(x,y)|-1$. Thus, by induction hypothesis, for all $s\in\adom{\chasece}$, there is at most one $u\in\adom{\chasece}$ such that $\chasece\models\cq'(s,u)$.

    Then, assume that there are some terms $s$, $t$ and $t'$ such that $\chasece\models\cq(s,t)\wedge\cq(s,t')$. Then, there is a unique $u$ such that $\chasece\models\cq'(s,u)$, so the atoms $\apred(t,u)$ and $\apred(t',u)$ are both present in $\chasece$. Let $\fulltrig$ be the trigger that introduces $u$ in $\chasece$. Since $\rsce$ is forward-existential, $t$ and $t'$ are both images of frontier terms of $\eru$. However, since $\rsce$ is frontier-functional, there is at most one $\apred$-atom in $\head{\eru}$. Thus, it must be that $t=t'$, which concludes the induction.

    The full result then follows from this case. Indeed, for all $y\in\vy$, consider a path from $y$ to $x$ in $\cq(x,\vy)$. This path defines a path query $\cq'_y(x,y)$, so for all $s\in\adom{\chasece}$, there is at most one $t\in\adom{\chasece}$ such that $\chasece\models\cq'_y(s,t)$. We then get the full result from the fact that for all $s$ and $\vt$, $\chasece\models \cq(s,\vt)$ implies $\chasece\models\cq'_y(s,t)$.
\end{proof}

\begin{proposition}
    If a valley query defines an  $\epred$-tournament of size $4$, it also defines an $\epred$-loop.
\end{proposition}

\begin{proof}
Let $q$ be a valley query.
\begin{itemize}
    \item If $q(x,y)$ is disconnected, then $q(x,y)=q_1(x)\wedge q_2(y)\wedge q_3$, with $q_i$ and $q_j$ having disjoint sets of variables if $i \not = j$. As there is an edge between $k_1$ and $k_2$, and they are symmetric, we can assume without loss of generality that $\chasece\models q_1(k_1)\wedge q_2(k_2)$. There is also an edge between $k_3$ and $k_4$, so we assume (again without loss of generality) that $\chasece\models q_1(k_3)\wedge q_2(k_4)$. Finally, there is an edge between $k_1$ and $k_3$, so either $\chasece\models q_2(k_1)$, or $\chasece\models q_2(k_3)$. In both cases, $\chasece\models\exists u~ q_1(u)\wedge q_2(u)\wedge q_3$, so $\chasece\models\exists u~ \epred(u,u)$, hence an $\epred$-loop is defined. We thus assume in all other cases that $q$ is connected.
    \item If $q(x,y)$ contains a single $\leq_q$-maximal vertex, say $x$ for instance (as the case where $y$ is $\leq_q$-maximal is symmetric). Thus, $y\sordc x$, so by \cref{lem:cq-functional}, the set $\set{\pair{s,t}\in \adom{\chasece}^2\mid \chasece\models q(s,t)}$ is a function. Thus, every vertex in $K$ has out-degree at most one, which contradicts the fact that $K$ is a tournament of size $4$.
    \item If both $x$ and $y$ are $\leq_q$-maximal in $q(x,y)$, then let $\vvv$ be the tuple of existential variables of $q$ that are smaller than both $x$ and $y$ (for $\leq_q$). Then, $q(x,y)=\exists \vvv~ q_x(x,\vvv)\wedge q_y(\vvv,y)$, where $q_x$ (resp. $q_y$) contains all the variables smaller than $x$ (resp. $y$). Then, by \cref{lem:cq-functional}, the sets $\set{\pair{x,\vvv}\mid\chasece\models q_x(x,\vvv)}$ and $\set{\pair{y,\vvv}\mid\chasece\models q_y(\vvv,y)}$ are functions, that we denote with $f_x$ and $f_y$, respectively. Since $K$ is a tournament over $4$ elements, it contains, up to renaming of elements, the following set: $\set{\epred(k_1,k_2),\epred(k_1,k_3),\epred(k_2,k_3)}$.

    We then show that $\chasec\models \epred(k_2,k_2)$. First, observe that if $q(s,t)\in\chasece$, then $f_x(s)=f_y(t)$. Thus, we get the three following equalities: $f_x(k_1)=f_y(k_2)$, $f_x(k_1)=f_y(k_3)$ and $f_x(k_2)=f_y(k_3)$. Thus, by composing these equalities, we get that $f_x(k_2)=f_y(k_2)$, so $\chasece\models q_x(k_2,f_x(k_2))\wedge q_y(f_x(k_2),k_2)$. Thus, $\chasece\models q(k_2,k_2)$, and $\chasec\models\epred(k_2,k_2)$, which concludes the proof.\qedhere

\end{itemize}
\end{proof}

\section{Discussion}

\subsection*{Tournament Definition}
While the main result focuses on tournaments over a fixed relation $\epred$, our results trivially extend to any relation definable by a binary UCQ. Consider an arbitrary binary UCQ $Q(x, y) = \bigvee_{i=1}^{k} q_i(x, y)$, where each $q_i$ is a CQ. One can simply add, for each $i$, the following rule to the considered rule set:  
\[ q_i(x, y) \to \epred(x, y). \]  

This addition does not affect the UCQ-rewritability of the rule set when $\epred$ is a fresh predicate symbol.

\subsection*{Arbitrary Colorability}
A natural next step in this line of research is to establish that with UCQ-rewritable rule sets, one cannot define structures of arbitrarily high chromatic number without entailing the $\loopq$ query:  
\begin{conjecture}\label{conj:colorability}
    For every signature $\signature$, every UCQ-rewritable rule set $\rs$ over $\signature$, and every instance $\instance$, we have:  
    \[\small
    \chase{\inst, \rs}|_{\epred} \text{ cannot be colored with a finite number of colors } \;\Rightarrow\; \chase{\inst, \rs} \models \loopq.
    \]
\end{conjecture}  

We believe that this conjecture, if proven, would constitute an elegant model-theoretic property of UCQ-rewritable rule sets and represent a significant step toward proving the \fcbdd conjecture. Let us elaborate. Similar to the results presented in this paper, if \cref{conj:colorability} holds, it would eliminate a substantial portion of the potential counterexample space: any structure that cannot be colored with a finite number of colors cannot be homomorphically embedded into a finite structure that does not entail $\loopq$.  

Our hope is that the tools developed throughout this paper will serve as a foundation for the proof of \cref{conj:colorability}. However, we note that such a proof would not be a straightforward extension of the techniques in \cref{sec:proof-core}. The current proof aims at showing existence of four-tournament witnessed by a single valley query. We know, however, that there exists structures of arbitrarily high chromatic numbers that do not contain the four-clique:
\begin{theorem}[Erdős \cite{Erdos_1959}]
    There exist graphs with arbitrarily high girth and chromatic number.
\end{theorem}  

\subsection*{Tournament Size Bounds}
Of separate, but related interest, is the study of the following question: 

\begin{question}\label{q:open-1}
    Given a UCQ-rewritable rule set $\rs$ such that $\chase{\topinst, \rs} \not\models \loopq$, what is the maximal $n$ such that $\chase{\inst, \rs}$ admits a tournament of size $n$?
\end{question}  


A careful reader will note that an upper bound on the maximal size of a tournament of the form $N(4, \ldots, 4)$ (with $|\qc|$ arguments) can be extracted from the proof presented in previous sections --- if a tournament of at least this size exists in the chase one can employ the machinery of \cref{sec:single-valley} to show that $\loopq$ is entailed.

We leave \cref{q:open-1} as an open question, as it ties together various significant problems regarding existential rules, particularly in the UCQ-rewritable fragment.

\section*{Acknowledgements}
	Piotr Ostropolski-Nalewaja was supported by grant 2022/45/B/ST6/00457 from the Polish National Science Centre (NCN).


\bibliographystyle{abbrv}
\bibliography{00-main}

\newpage
\appendix
\section{Proof of Lemma~\ref{lem:streamline-prop-one}}\label{app:streamline-prop-one}
\newcommand{\streaminit}{\stream_\mathrm{init}(\rsb)}
\newcommand{\streamex}{\stream_\exists(\rsb)}
\newcommand{\streamDL}{\stream_\mathrm{DL}(\rsb)}

\newcommand{\siga}{\sig_{\apred}}
\newcommand{\sigb}{\sig_{\bpred}}

Let $\streaminit$ be the subset of $\stream(\rsb)$ containing only rules of form $\eruinit$, let $\streamex$ and $\streamDL$ be defined in the analogous way.
\begin{observation}\label{obs:stream-signature-tripartition}
    For any rule set $\rsb$, the signature of $\stream(\rsb)$ can be partitioned into three disjoint signatures $\sig$, $\siga$ and $\sigb$, such that 
    \begin{itemize}
		\item $\sig$ is the signature of heads of rules in $\streamDL$ and bodies of rules in $\streaminit$,
        \item $\siga$ is the signature of heads of rules in $\streaminit$ and bodies of rules in $\streamex$,
        \item $\sigb$ is the signature of heads of rules in $\streamex$ and bodies of rules in $\streamDL$,
    \end{itemize}
\end{observation}

\newcommand{\homtostream}[1]{\sigma^\stream_{#1}}
\begin{lemma}\label{lem:streamline-prop-two-lem-one}
    For every rule set $\rsb$, every instance $\jnst$ and every integer $i$, there is a homomorphism $\homtostream{i}$ from $\step{i}{\jnst,\rsb}$ to $\step{3i}{\jnst,\stream(\rsb)}$.
\end{lemma}
\begin{proof}
    We proceed by induction on $i$. For the base case, we have $\step{0}{\jnst,\rsb}=\step{0}{\jnst,\stream(\rsb)}=\jnst$, so $\homtostream{0}$ is the identity. Now, for the induction step, assume $\homtostream{i}$ constructed, and consider a trigger $\trig=\fulltrig$ over $\step{i}{\jnst,\rsb}$. Then, $\homtostream{i}(\hom(\body{\eruinit}))=\homtostream{i}(\hom(\body{\eru}))$ (as these bodies are equal), so the trigger $\pair{\eruinit, \homtostream{i}\circ\hom}$ is a trigger over $\step{3i}{\jnst,\stream(\rsb)}$. Thus, the trigger $\pair{\eruex,\homtostream{i}\circ\hom}$ is a trigger over $\step{3i+1}{\jnst,\stream(\rsb)}$, as $\eruex$'s body is $\eruinit$'s head. The same applies for $\pair{\eruDL,\homtostream{i}\circ\hom}$ at step $3i+2$. Thus, there is an isomorphism between $\homtostream{i}(\trigoutput{\trig})$ and $\trigoutput{\pair{\eruDL,\homtostream{i}\circ\hom}}$, so we can extend $\homtostream{i}$ using these isomorphisms for all the triggers applied at step $i+1$ in $\step{i+1}{\jnst,\rsb}$, constructing $\homtostream{i+1}$ as a homomorphism from $\step{i+1}{\jnst,\rsb}$ to $\step{3i+3}{\jnst,\stream(\rsb)}$.
\end{proof}

We denote the signature of $\rsb$ with $\sig$, and define for every instance $\jnst$ the restriction $\jnst_{|\sig}$ of $\jnst$ to the signature $\sig$.
\begin{observation}\label{obs:stream-body-to-full-head}
    For every instance $\jnst$, rule set $\rsb$, rule $\eru\in\rsb$ and homomorphism $\hom$
    \begin{itemize}
        \item If $\pair{\eruex,\hom}$ is a trigger such that some atom in $\hom(\body{\eruex})$ is mapped to $\step{1}{\jnst,\streaminit}\setminus\jnst$, then there is some trigger $\pair{\eruinit,\hom'}$ for $\jnst$ such that $\trigoutput{\pair{\eruinit,\hom'}}=\hom(\body{\eruex})$.
        \item If $\pair{\eruDL,\hom}$ is a trigger such that some atom in $\hom(\body{\eruDL})$ is mapped to $\step{1}{\jnst,\streamex}\setminus\jnst$, then there is some trigger $\pair{\eruex,\hom'}$ for $\jnst$ such that $\trigoutput{\pair{\eruex,\hom'}}=\hom(\body{\eruDL})$.
    \end{itemize}
\end{observation}
\begin{proof}
	Notice that all the atoms in the body of these rules feature a term that is an existential variable in the head of the previous rule. Thus, this variable has to be generated by a rule application. Since for all $\eru\in\rsb$, $\eruinit$'s head is $\eruex$'s body, and the same for $\eruex$ and $\eruDL$, we get the result.
\end{proof}

\newcommand{\homfromstream}[1]{\pi^\stream_{#1}}
\newcommand{\unfoldedchase}[1]{\mathcal{R}_{#1}}
\begin{lemma}\label{lem:streamline-prop-two-lem-two}
    For every rule set $\rsb$ and instance $\jnst$, let $\jnst'=\step{1}{\step{1}{\jnst,\streamex}, \streamDL}_{|\sig}$. For all $i$, there is a homomorphism $\homfromstream{i}$ from $\step{3i}{\jnst,\stream(\rsb)}$ to
    \[\unfoldedchase{i}=\step{1}{\step{1}{\step{i}{\jnst', \rsb}, \streaminit}, \streamex}\cup\step{1}{(\jnst\setminus\jnst_{|\sig}), \streamex}\]
\end{lemma}
\begin{proof}
    We construct $\homfromstream{i}$ by induction on $i$. For the base case, $\step{0}{\jnst,\stream(\rsb)}=\jnst$ is included in $\unfoldedchase{0}$, as $\jnst_{|\sig}\subseteq\jnst'$, so $\homfromstream{0}$ is the identity. We then focus on the inductive case, by assuming that $\homfromstream{i}$ is constructed. Triggers generating atoms in $\step{3i+3}{\jnst,\stream(\rsb)}\setminus\step{3i}{\jnst,\stream(\rsb)}$ can be applied at step $3i+1$, $3i+2$ or $3i+3$, and use a rule in $\streaminit$, $\streamex$ or $\streamDL$. Let $\trig=\pair{r,\hom}$ be such a trigger. To construct $\homfromstream{i+1}$, we show in these nine cases how $\homfromstream{i}(\trigoutput{\trig})$ can be mapped to $\unfoldedchase{i+1}$.
    \begin{description}
        \item[Step $3i+1$ and $r=\eruinit$:] The body of these rules is over $\sig$, meaning that $\pair{r,\homfromstream{i}\circ\hom}$ is a trigger over $\step{i}{\jnst', \rsb}$. Thus, $\homfromstream{i}(\trigoutput{\pair{r,\hom}})$ maps into $\trigoutput{\pair{r,\homfromstream{i}\circ\hom}}\subseteq\step{1}{\step{i}{\jnst', \rsb}, \streaminit}\subseteq\unfoldedchase{i}$.

        \item[Step $3i+1$ and $r=\eruex$:] The body of these rules is over $\siga$, so it can only be mapped to atoms in $\jnst\setminus\jnst_{|\sig}$ or in $\step{1}{\step{i}{\jnst', \rsb}, \streaminit}$. If $\homfromstream{i}(\hom(\body{\eruex}))$ is included in $\jnst$, then $\homfromstream{i}(\trigoutput{\pair{\eruex,\hom}})$ can be mapped to $\step{1}{(\jnst\setminus\jnst_{|\sig}), \streamex}$. Otherwise, by \cref{obs:stream-body-to-full-head}, there is some trigger $\pair{\eruinit,\hom'}$ such that $\homfromstream{i}(\hom(\body{\eruex}))=\trigoutput{\pair{\eruinit,\hom'}}$. Thus, $\homfromstream{i}(\trigoutput{\pair{\eruex,\hom}})$ can be mapped into $\trigoutput{\pair{\eruex,\homfromstream{i}\circ\hom}}\subseteq\unfoldedchase{i}$.

        \item[Step $3i+1$ and $r=\eruDL$:] If $\homfromstream{i}(\hom(\body{\eruDL}))\subseteq\jnst\setminus\jnst_{|\sig}$, then $\homfromstream{i}(\trigoutput{\pair{\eruDL,\hom}})$ can be mapped into $\trigoutput{\pair{\eruDL,\homfromstream{i}\circ\hom}}\subseteq\jnst'\subseteq\unfoldedchase{i}$, as rules in $\streamDL$ only produce atoms over $\sig$.
        Otherwise, by applying \cref{obs:stream-body-to-full-head} twice, there are some triggers $\pair{\eruex,\hom'}$ over $\step{1}{\step{i}{\jnst', \rsb}, \streaminit}$ and $\pair{\eruinit,\hom''}$ over $\step{i}{\jnst', \rsb}$ such that $\trigoutput{\pair{\eruinit,\hom''}}=\hom'(\body{\eruex})$ and $\trigoutput{\pair{\eruex,\hom'}}=\homfromstream{i}(\hom(\body{\eruDL}))$. Thus, the trigger $\pair{\eru,\hom''}$ is a trigger over $\step{i}{\jnst', \rsb}$ such that $\trigoutput{\pair{\eru,\hom''}}$ and $\trigoutput{\pair{\eruDL,\homfromstream{i}\circ\hom}}$ are isomorphic. Thus, we can extend this isomorphism to map the set $\homfromstream{i}(\trigoutput{\pair{\eruDL,\hom}})$ into $\trigoutput{\pair{\eru,\hom''}}\subseteq \step{i+1}{\jnst', \rsb}\subseteq \unfoldedchase{i+1}$.

        \item[Step $3i+2$ and $r=\eruinit$:] Some atom of $\hom(\body{\eruinit})$ has to be mapped into a newly created atom at step $3i+1$ by a rule in $\streamDL$. For all such trigger $\pair{\eruDL',\hom'}$, we can map $\homfromstream{i}(\trigoutput{\pair{\eruDL',\hom'}})$ into $\step{i+1}{\jnst', \rsb}$ by the previous case. Thus, by \cref{obs:stream-body-to-full-head}, we can map $\homfromstream{i}(\hom(\body{\eruinit}))$ entirely into $\step{i+1}{\jnst', \rsb}$, so that $\pair{\eruinit,\homfromstream{i}\circ\hom}$ is a trigger for it, and thus $\homfromstream{i}(\trigoutput{\pair{\eruinit,\hom}})$ can be mapped into $\trigoutput{\pair{\eruinit,\homfromstream{i}\circ\hom}}\subseteq \unfoldedchase{i+1}$.

        \item[Step $3i+2$ and $r\in\set{\eruex,\eruDL}$:] These cases are treated exactly as step $3i+1$.
        
		\item[Step $3i+3$ and $r\in\set{\eruinit,\eruDL}$:] These cases are treated exactly as step $3i+2$.
        
		\item[Step $3i+3$ and $r=\eruex$:] Some atom of $\hom(\body{\eruex})$ has to be mapped into a newly created atom at step $3i+2$ by a rule in $\streaminit$. For all such trigger $\pair{\eruinit',\hom'}$, we can map $\homfromstream{i}(\trigoutput{\pair{\eruinit',\hom'}})$ into $\step{1}{\step{i+1}{\jnst', \rsb},\streaminit}$ by the previous cases. Thus, by \cref{obs:stream-body-to-full-head}, we can map $\homfromstream{i}(\hom(\body{\eruex}))$ entirely into $\step{1}{\step{i+1}{\jnst', \rsb},\streaminit}$, so that $\pair{\eruex,\homfromstream{i}\circ\hom}$ is a trigger for it, and thus $\homfromstream{i}(\trigoutput{\pair{\eruex,\hom}})$ can be mapped into $\trigoutput{\pair{\eruex,\homfromstream{i}\circ\hom}}\subseteq \unfoldedchase{i+1}$.\qedhere
    \end{description}
\end{proof}

\begin{customlem}{\ref{lem:streamline-prop-one}}[restated]
	For every rule signature $\sig$, rule set $\rsb$ and instance $\jnst$ over $\sig$, we have that $\chase{\jnst, \rsb}$ is homomorphically equivalent to $\chase{\jnst, \stream(\rsb)}$ when restricted to signature of $\rsb$.
\end{customlem}
\begin{proof}
	We get this result as a corollary of \cref{lem:streamline-prop-two-lem-one} and \cref{lem:streamline-prop-two-lem-two}, by noticing that
	\[\left(\step{1}{\step{1}{\step{i}{\jnst', \rsb}, \streaminit}, \streamex}\cup(\jnst\setminus\jnst_{|\sig})\right)_{|\sig}=\step{i}{\jnst', \rsb}\]
	and $\jnst'=\jnst$ when $\jnst$ is over $\sig$.
\end{proof}

\newcommand{\erufull}{\eru_{\mathrm{full}}}
\newcommand{\rsbfull}{\rsb_{\mathrm{full}}}
\section{Proof of Lemma~\ref{lem:streamline-prop-two}}\label{app:C}
\begin{customlem}{\ref{lem:streamline-prop-two}}[restated]
Given a rule set $\rsb$ we have that $\stream(\rsb)$ is forward-existential and predicate-unique. Moreover, if $\rsb$ is UCQ-rewritable then $\stream(\rsb)$ is UCQ-rewritable as well.
\end{customlem}

It should be clear from the definition of $\stream(\rsb)$ that it is forward-existential and predicate unique. We show now that it is UCQ-rewritable as well.
Let $\streamex$ be $\set{\eruex \mid \eru \in \rsb}$ and $\streamDL$ be $\set{\eruDL \mid \eru \in \rsb}$.
For $\eru \in \rsb$ define the rule $\erufull$ as:
\[\erufull \;\;=\quad\quad \narule{\phi(\vx, \vy)}{\vz,w}{\psi(\vy, \vz) \quad\wedge\quad \apred_{0}^\eru(w)\wedge\bigwedge_{y\in\vy} \apred^\eru_{y,w}(y,w) \quad\wedge\quad \bigwedge_{z\in\vz}\bpred}^\eru_{y',z}(y',z)\]
and let $\rsbfull$ be the set $\set{\erufull \mid \eru \in \rsb}$ and $\rsb' = \rsbfull \cup \streamex \cup \streamDL$.
It should be clear that for any instance $\inst$:
\[\chase{\inst, \rsb'} \homequiv \chase{\inst, \stream(\rsb)}.\]
Due to the above, the rewritings against $\rsb'$ and $\stream(\rsb)$ are equivalent --- so $\rsb'$ is UCQ-rewritable \iffi $\stream(\rsb)$ is --- therefore it is enough to show that $\rsb'$ is UCQ-rewritable. Now, observe the following:
\[\chase{\inst, \rsb'} \homequiv \chase{\;\;\chase{\inst,\; \streamex \cup \streamDL},\;\; \rsbfull}.\]
We argue it holds, as $\streamex \cup \streamDL$ is non-recursive and applications of rules from $\rsbfull$ cannot trigger rules from $\streamex \cup \streamDL$ --- the heads of their rules are already contained in the respective heads of rules of $\rsbfull$. 

Due to the above equivalence, we can argue the $\rsb'$ is UCQ-rewritable if $\rsbfull$ is. It holds as given a CQ $q$ and its \bdd constant $c$ under $\rsbfull$ the \bdd constant of $q$ under $\rsb'$ is no greater than $c + 2$, as $\streamex \cup \streamDL$ is not recursive, and all queries under $\streamex \cup \streamDL$ have the \bdd constant of at most $2$. 

Therefore, by showing $\rsbfull$ is UCQ-rewritable we would complete the proof. Let $\sig'$ contain the fresh predicates $\apred$ or $\bpred$ (with indices) appearing in $\rsbfull$. Note that predicates from  $\sig'$ are not used in $\rsbfull$ recursively and that when restricted to predicates from $\sig$ the rule sets $\rsbfull$ and $\rsb$ are identical. Therefore the derivations of $\rsbfull$ and $\rsb$ are identical --- from this we conclude that $\rsbfull$ is indeed UCQ-rewritable as $\rsb$ is UCQ-rewritable.\qed

\end{document}